\documentclass[conference]{IEEEtran}
\IEEEoverridecommandlockouts

\newtheorem{prop}{Proposition}
\newtheorem{example}{Example}
\newtheorem{rmk}{Remark}
\usepackage{cite}

\ifCLASSINFOpdf
  \usepackage[pdftex]{graphicx}
  \DeclareGraphicsExtensions{.pdf,.jpeg,.png}
\else
   \usepackage[dvips]{graphicx}
   \DeclareGraphicsExtensions{.eps}
\fi
\usepackage[cmex10]{amsmath}
\usepackage{algorithm,algpseudocode}
\usepackage{amsfonts}
\usepackage{bm}
\usepackage[font=footnotesize]{subfig}
\begin{document}
%
\title{Decoding $q$-ary lattices in the Lee metric}

\author{\IEEEauthorblockN{Antonio C. de A. Campello Jr.\IEEEauthorrefmark{1}}
\IEEEauthorblockA{Institute of Mathematics, Statistics\\and Computer Science\\
University of Campinas, S\~ao Paulo\\
13083-859, Brazil\\
Email: campello@ime.unicamp.br} \and \IEEEauthorblockN{Grasiele C.
Jorge\IEEEauthorrefmark{2}}
\IEEEauthorblockA{Institute of Mathematics, Statistics\\and Computer Science\\
University of Campinas, S\~ao Paulo\\
13083-859, Brazil\\
Email: grajorge@ime.unicamp.br} \and \IEEEauthorblockN{Sueli I. R.
Costa\IEEEauthorrefmark{3}}\thanks{Work partially supported by
FAPESP\IEEEauthorrefmark{1} under grant 2009/18337-6,
CNPq\IEEEauthorrefmark{2} 140239/2009-0, CNPq\IEEEauthorrefmark{3}
309561/2009-4 and  FAPESP 2007/56052-8}
\IEEEauthorblockA{Institute of Mathematics, Statistics\\and Computer Science\\
University of Campinas, S\~ao Paulo\\
13083-859, Brazil\\
Email: sueli@ime.unicamp.br}}


%


\maketitle

\begin{abstract}
$q$-ary lattices can be obtained from $q$-ary codes using the
so-called Construction A. We investigate these lattices in the Lee
metric and show how their decoding process can be related to the
associated codes. For prime $q$ we derive a Lee sphere decoding
algorithm for $q$-ary lattices, present a brief discussion on its
complexity and some comparisons with the classic sphere decoding.
\end{abstract}


%

\section{Introduction}
A $q$-ary lattice \cite{Mic,Superball} is an integer lattice in the
Euclidean space $\mathbb{R}^{n}$ which contains $q\mathbb{Z}^{n}$ as
a sublattice. It can be obtained via Construction A \cite{Sloane}
from a linear code in the module $\mathbb{Z}_q^{n}$. Those lattices
have deserved special attention in recent years due to their use in
cryptographic schemes based on lattices, one of the so-called
``post-quantum'' methods \cite{Mic,Mic2}. One important problem
concerning general lattices (thus particularly $q$-ary lattices) is
the CVP (Closest Vector Problem) which asks for the closest lattice
point to a received point in $\mathbb{R}^n$. A method largely used
to solve this problem in the Euclidean metric is the sphere decoding
\cite{Pohst, Viterbo1, Viterbo2}, which has exponential expected
complexity \cite{Babak}. Other methods include basis reductions such
as LLL \cite{LLL} and BKZ \cite{BKZ}, and trellis algorithm
\cite{Banihashemi}.

Codes in the Lee metric, on the other hand, were introduced in
\cite{Lee} and since then have been the object of study of many
works from both theoretical (e.g. \cite{Peter, Golomb, Bader,
Vardy}) and practical (e.g. \cite{Roth, Etzion}) points of view. The
Lee metric has a close relation with the $l_1$ metric, also called
Manhattan or Taxi Cab metric, explored for example \cite{Golomb,
Bader}, concerning the existence of perfect codes and more recently
\cite{Vardy}, where the authors show how to construct dense
error-correcting codes in the Lee metric from dense lattice packings
of $n$-dimensional cross-polytopes.

The contributions in this paper are organized as follows. In Section
III we derive connections between a $q$-ary lattice and its
associated code decoding processes in the Lee metric through
Propositions \eqref{Prop1} and \eqref{Prop2}. This illustrates the
fact that the Lee metric seems to have a ``natural'' geometry when
dealing with $q$-ary codes (and lattices) for $q \in \mathbb{N},$ $q
> 3$. In Section IV we propose an
adaptation of the traditional sphere decoding ideas for $q$-ary
lattices ($q$ prime) in the Lee metric and discuss its expected
complexity through arguments which are similar to the ones presented
in \cite{Babak}. Using geometric arguments, we also make some
comparisons with the classic sphere decoding and perform some low
dimensional simulations.

\section{Preliminaries}
In this section we summarize some concepts and properties related to $q$-ary lattices, Lee metric and establish
the notation to be used from now on.
\subsection{$q$-ary lattices}
 Given $q \in \mathbb{N}$, a \textit{$q$-ary linear code $C$} is a $\mathbb{Z}_q$-submodule of
$\mathbb{Z}_q^n$. For prime $q$, there is always a generator matrix
for $C$ in the \textit{systematic form}, i.e., $A_{n \times k} \sim
[I_{k \times k} \,\,  | \,\, B_{k \times n-k}^t]^t.$

A \textit{lattice} $\Lambda$ is a discrete additive subgroup of
$\mathbb{R}^n$. Equivalently, $\Lambda \subseteq \mathbb{R}^n$ is a
lattice iff there are linearly independent vectors ${\bm
v_1},\ldots,{\bm v_m} \in \mathbb{R}^n\mbox{, }$ such that any $y
\in \Lambda$ can be written as $y = \sum_{i=1}^m \alpha_i {\bm v_i}
\,\, \mbox{,} \,\, \alpha_i \in \mathbb{Z}$. The set $\{{\bm
v_1},\ldots,{\bm v_m}\}$ is called a \textit{basis} for $\Lambda$. A
matrix $M$ whose columns are these vectors is said to be a
\textit{generator matrix} for $\Lambda$. Given a metric $d$ in
$\mathbb{R}^n$, the Voronoi region of ${\bm x} \in \Lambda$ is the
set $V(x) = \{{\bm y} \in \mathbb{R}^{n}; d({\bm y},{\bm x}) \leq
d({\bm y},{\bm x^{*}}),\,\, \mbox{for}\,\, \mbox{all}\,\, {\bm
x^{*}} \in \Lambda\}$. To decode ${\bm y} \in \mathbb{R}^n$ is to
find the closest lattice point to ${\bm y}$.

The so-called Construction A extended for $q$-ary codes
\cite{Superball}, can be described by the surjective map $\phi:
\mathbb{Z}^n \longrightarrow \mathbb{Z}_q^n$, $\phi(x_1,\ldots,x_n)
= (\overline{x_1},\ldots,\overline{x_n}).$ Given a linear code $C
\subseteq \mathbb{Z}_q^n$, $\Lambda_q (C) = \phi^{-1}(C)$ is said to
be the $q$-ary lattice associated to $C$ and
$\Lambda_q(C)/q\mathbb{Z}^n \approx C$. The code $C$ can be viewed
as the set of representatives of the above quotient inside the
hypercube $\left[0,q \right)^n$ and $\Lambda_q(C)$ is given by
translations of this set by multiples of $q$ in each direction. For
$q$ prime, $\Lambda_q(C)$ is generated by the matrix:

\begin{equation}
M=\left[\begin{array}{cc}
I_{k\times k} & 0_{k\times(n-k)}\\
B_{(n-k) \times k} & qI_{(n-k)\times(n-k)}\end{array}\right]
\label{eq:Geradora}
\end{equation}

\noindent provided that $[I_{k \times k}  \,\,  | \,\, B_{k \times
n-k}^t]^t$ is the associated generator matrix for $C$.

\subsection{Lee metric}

Instead of the usual Hamming metric for codes and Euclidean for
lattices we consider here the Lee metric for both spaces which seems
to be more natural when dealing with $q$-ary lattices and codes.

For ${\bm x}=(x_1,\ldots,x_n), {\bm y}=(y_1,\ldots,y_n) \in
\mathbb{R}^{n},$ the $l_1$ or sum distance is defined as
$d_{l_1}({\bm x},{\bm y}) = \displaystyle \sum_{i=1}^n|x_i-y_i|.$
The Lee distance in either
$\mathbb{Z}_{q}^{n}=\mathbb{Z}/q\mathbb{Z}^{n}$ or
$\mathbb{R}^{n}/q\mathbb{Z}^{n}$ is the distance induced by
$d_{l_1}$ through the quotient map: $d_{Lee}(\overline{\bm
x},\overline{\bm y}) = \displaystyle \sum_{i=1}^n \min
\left\{x_i-y_i \,\, (\mbox{mod }q), \,\, y_i-x_i \,\, (\mbox{mod }q) \right\}.$ We
will denote here either $d_{l_1}$ or $d_{Lee}$ by $d$ and call both
Lee distance. The \textit{minimum norm} $\mu$ of a lattice $\Lambda$
is $\mu = \displaystyle\min_{{\bm 0} \neq {\bm x} \in \Lambda}
d({\bm x},{\bm 0})$ and for a $q$-ary lattice $\Lambda_q (C)$ we
have $\mu = \min \{q, d(C)\}$ \cite{Superball}.

\section{Decoding $q$-ary lattices via Construction A}

A decoding process for lattices
constructed from binary codes via Construction A is presented in \cite{Sloane}. It is shown that
decoding a binary code $C \subseteq \mathbb{Z}_{q}^{n}$  corresponds to
decoding  in the binary lattice $\Lambda_{2}(C)\subseteq \mathbb{R}^{2}$ in the Euclidean
metric. In this section we obtain the same kind of relation between code
and lattice decoding with the Lee metric.

Let $C \subseteq \mathbb{Z}_{q}^{n}$ be a $q$-ary code. Due to the
isomorphism $\Lambda_{q}(C)/q\mathbb{Z}^{n} \simeq C$, we will not
distinguish  the elements of $\Lambda_{q}(C)/q\mathbb{Z}^{n}$ from
codewords of $C$ and we will denote by $\lceil  \rfloor$ the
rounding to the nearest integer. Given a received vector ${\bm r}
\in \mathbb{R}^n,$ let ${\bm z}$ be its closest point in
$\Lambda_{q}(C)$ considering the Lee metric. In the next
Propositions \eqref{Prop1} and \eqref{Prop2},  we show how to find
via Construction A a representative of ${\overline{\bm z}}$ which is
given by a codeword in $C.$

\begin{prop} \label{Prop1} Let $\Lambda_{q}(C)$ be a $q$-ary lattice and ${\bm r} = (r_1,\ldots,r_n)^{t}
\in \mathbb{R}^n$ a received vector. Given an element
${\overline{\bm x}} \in \Lambda_{q}(C)/q\mathbb{Z}^{n},$ ${\bm x} =
({x_1},\ldots,{x_n})^{t},$ the representative $\bm{z} =
(z_1,\ldots,z_n)^{t} \in \Lambda_q(C)$ of $\bm{\bar{x}}$ which is
closest to ${\bm r}$ in $\Lambda_{q}(C)$ considering the Lee metric
is given by $z_i = x_i + qw_i$ where $w_i = \left
\lceil\displaystyle\frac{r_i - x_i}{q}\right\rfloor,$ for each
$i=1,\ldots,n.$
\end{prop}
\begin{proof} The proof is straightforward. A representative of the class of ${\bm x}$ is
given by ${\bm z} = {\bm x} + q {\bm w}$, where ${\bm w} \in
\mathbb{Z}^n$ and the Lee distance $d({\bm r},{\bm z}) =
\sum_{i=1}^{n}|r_i - x_i - q w_i|$ is minimum when $w_i =\left\lceil
\frac{r_i - x_i}{q}\right \rfloor.$
\end{proof}

\begin{prop} \label{Prop2} Let $\Lambda_{q}(C)$ be a $q$-ary lattice. Given ${\bm r}=(r_1,\ldots,r_n)^{t}
\in \mathbb{R}^n$ a
received vector let ${\bm r}\,\, (\mbox{mod }q) \in [0,q)^n$, obtained
from ${\bm r}$ by reductions modulo $q$  in each entry. If
${\overline{\bm x} } \in C$ is an element of $C$ closest to ${\bm
r}\,\, (\mbox{mod }q)$ considering the Lee metric, ${\bm z} \in
\Lambda_{q}(C),$ ${\overline{\bm z}} = {\overline{\bm x}},$ given by
Proposition \eqref{Prop1} is a lattice point nearest to ${\bm r}$.
\end{prop}

\begin{proof} Let ${\bm r}=(r_1,\ldots,r_n)^{t} =
({r_1}^{*},\ldots,{r_n}^{*})^{t} + q (t_1,\ldots,t_n)^{t},$ with $0
\leq {r_i}^{*} \leq q,$ $t_i \in \mathbb{Z}$, for $i=1,\ldots,n$,
that is ${\bm r} \,\, (\mbox{mod }q)  = ({r_1}^{*},\ldots,{r_n}^{*})^{t}.$
Let ${\bm {\overline{x}}} \in C,$ ${\bm x}=
({x_1},\ldots,{x_n})^{t}$, $0 \leq x_i \leq q-1$ for $i=1,\ldots,n,$
be a closest point to ${\bm r}\,\, (\mbox{mod }q)$ considering the Lee
metric. We will show that a closest point to ${\bm r}$ in $\Lambda$
is in the same class that ${\bm x}$ in
$\Lambda_{q}(C)/q\mathbb{Z}^{n}.$ For each class ${\overline{\bm a}}
\in C,$ ${\bm a} = (a_1,\ldots,a_n)^{t},$  by Proposition
\eqref{Prop1} we find the representative ${\bm a^{*}}$ closest to
${\bm r}$ considering the Lee  metric. We will show that $d({\bm
r},{\bm a^{*}}) = d({\bm r}\,\, (\mbox{mod }q), \,\, {\overline{\bm a}}).$ For
the Lee distance we have
$$d({\bm r},{\bm a^{*}}) =  \sum_{i=1}^{n} |{r_i}^{*} - a_i - \alpha_i q |,$$ where
$\alpha_i = \left( \left\lceil\frac{{r_i}^{*} - a_i}{q} + t_i
\right\rfloor - t_i \right) .$ Since $-1 \leq \frac{{r_i}^{*} -
a_i}{q} \leq 1$ because $|{r_i}^{*}-a_i| \leq q,$ then $\alpha_i \in
\{-1,0,1\}.$  Now, we can observe that if $\alpha_i = 0$ for some
$i,$ then $-q/2 \leq {r_i}^{*} - a_i \leq q/2$ and this implies
$\min\{|r_i^{*}-a_i|,q-|r_i^{*}-a_i|\} = |r_i^{*}-a_i|.$ If $\alpha_i
= 1$ for some $i,$ then $q/2 < {r_i}^{*} - a_i \leq q $ and then
$\min\{|r_i^{*}-a_i|,q-|r_i^{*}-a_i|\} = q -|r_i^{*}-a_i| \mbox { and
} |r_i^{*}-a_i| = r_i^{*}-a_i.$ Finally, if $\alpha_i = -1$ for some
$i$, then $-q \leq {r_i}^{*} - a_i < -q/2$ and then
$\min\{|r_i^{*}-a_i|,q-|r_i^{*}-a_i|\} = q - |r_i^{*}-a_i| \mbox{ and
} |r_i^{*}-a_i| = -(r_i^{*}-a_i).$ So,  $d({\bm r},{\bm a^{*}}) =
\sum_{i=1}^{n} |{r_{i}}^{*} - a_i  - \alpha_i q| = \sum_{i=1}^{n}
\min\{|{r_i}^{*} - a_i|, q-|{r_i}^{*}-a_i|\} = d({\bm
r}\,\, (\mbox{mod }q), \,\, {\overline{\bm a}}).$ Since ${\overline{\bm x}}$
satisfies $d({\bm r}\,\, (\mbox{mod }q), \,\,  {\overline{\bm x}}) = \min
\{d({\bm r}\,\, (\mbox{mod }q), \,\,  {\overline{\bm a}}), {\overline{\bm a}}
\in C\}$ then ${\bm z}$ satisfies $d({\bm r},{\bm z})= \min\{d({\bm
r},{\bm{y}}),{\bm y} \in \Lambda_{q}(C)\}.$ \end{proof}

\begin{example} Consider the cyclic $13$-ary code in $\mathbb{Z}_{13}^{2},$
$C = \left< (\overline{1},\overline{5})^{t}\right>.$ It has minimum
Lee distance  $d(C)=5$ and error correction capacity $t=2$. For the
received vector ${\bm r} =(0,-6)^{t}$ the Lee-closest codeword to
${\bm r}\,\, (\mbox{mod } q)=(\overline{0},\overline{7})^{t}$ is
${\bm x}=(\overline{12},\overline{8})^{t}.$ Hence in Proposition
\eqref{Prop1}, $w_1 = w_2 = -1$ and by Proposition \eqref{Prop2} the
closest lattice point to ${\bm r}$ is ${\bm z} = (-1,-5)^{t}.$
Figure (1) shows the lattice $\Lambda_{13}(C)$ and its Voronoi
regions.
\end{example}

\begin{figure}[!t]
\centering
\includegraphics[scale=0.7]{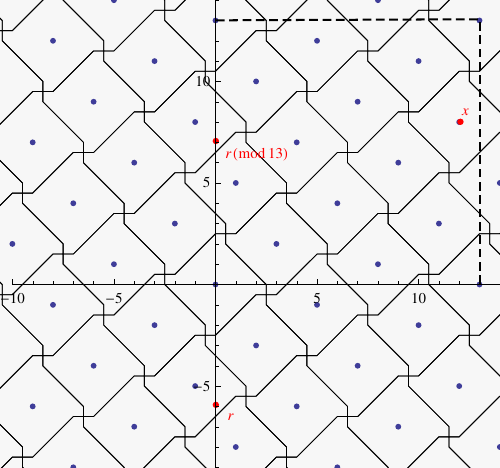}
\caption{Decoding  $\bm{r}=(0,-6) \in \mathbb{R}^2$}
\end{figure}

\begin{example} Consider ${\bm C}$ the BCH code defined in the ring $\frac{\mathbb{Z}_4[x]}{<f(x)>}, $ where $f(x) =
x^3+x+1$ with parity-check matrix
\[\left(\begin{array}{ccccccc} 1 & 1 & 1 & 1 & 1 & 1 & 1 \\ \alpha^5 & \alpha &  1 & \alpha^3 &
 \alpha^6 & \alpha^2 & \alpha^4 \end{array}\right)^{t},\] where $\alpha=\beta^{2}$ and $\beta$ is a root of $f(x)$
 \cite{Andrade}. We can derive a generator matrix for the
 $4$-ary lattice $\Lambda=\phi^{-1}({\bm C})$ as
\[\left(\begin{array}{ccccccc} 1 & 0 & 0 & 0 & 0 & 0& 0\\ 0 & 1 &  0 & 0 &
 0 & 0 & 0 \\  0 & 0 & 1 & 0 & 0 & 0 & 0 \\ 2 & 1 & 3 & 4 & 0 & 0 & 0 \\ 1 & 3 & 2 & 0 & 4 & 0 & 0 \\
 1 & 1 & 1 & 0 & 0 & 4 & 0 \\
 3 & 2 & 1 & 0 & 0 & 0 & 4 \end{array}\right).\] Let ${\bm r} = (0,7,4,8,0,12,0)^{t}$  be a received vector.
 The closest code point to  ${\bm r}\,\, (\mbox{mod }q) $ $=
 (\overline{0},\overline{3},\overline{0},\overline{0},\overline{0},\overline{0},\overline{0})^{t}$ is
 $\overline{\bm x} = (\overline{0},\overline{0},\overline{0},\overline{0},\overline{0},\overline{0},\overline{0})^{t}.$
 Hence in Proposition \eqref{Prop1},  $w_1 =0$, $w_2=\left[\frac{7}{4}\right] = 2,$
 $w_3 =\left[\frac{4}{4}\right]=1$, $w_4 =
 \left[\frac{8}{4}\right]=2$, $w_5=0$,
 $w_6=\left[\frac{12}{4}\right]=3$ and $w_7=0.$ Then, ${\bm z}= (0,0,0,0,0,0,0)^{t} + 4(0,2,1,2,0,3,0)^{t} =
 (0,8,4,8,0,12,0)^{t}$ is the closest point to ${\bm r}.$
 \label{exemplo:BCH}
 \end{example}

Proposition (\ref{Prop2}) provides a decoding process for $q$-ary lattices with the Lee metric via its generator code.
This can be specially interesting for associated codes with an efficient Lee decoding algorithm.  Decoding algorithms
for some $q$-ary codes in the Lee metric assuming integer coordinates for the received point ${\bm r}$  can be found
in \cite{Bader,Peter,Roth}. The algorithm derived in the next section is based only on the lattice structure and allows
real coordinates for ${\bm r}$.

\section{Lee Sphere Decoding}

The algorithm proposed here is analogous to the classic sphere
decoding in the Euclidean metric and follows the same basic ideas.
Nevertheless, we show that the structure of $q$-ary lattices in the Lee metric yields some important simplifications in
comparison to the traditional algorithm. From now on, $\left\| .
\right\|$ will always stand for the Lee norm.

Let $\Lambda_q (C)$ be a $q$-ary lattice with generator matrix $M$ in the special form given in
\eqref{eq:Geradora} and $\bm{r}$ a received point. We remark that $\Lambda_q(C)$ always have a generator matrix in that form for $q$ prime and in some cases for $q$ not prime, as can be seen in Example \eqref{exemplo:BCH}. Given $R > 0$ we
want to enumerate all $\bm{y} \in \Lambda_q (C)$ such that
$\left\|\bm{y} - \bm{r} \right\| = \left\| M\bm{x} - \bm{r} \right\|
\leq R$ and then find the closest lattice
point to $\bm{r}$.

Fixing the vector $\bm{x^1} = (x_1,\ldots,x_k)^t$, the minimum of
$\left\|M\bm{x} - \bm{r} \right\|$ is obtained by simply taking
\begin{equation} \label{minimizadores} x_j = \left\lceil (r_j - (B\bm{x^1})_{j-k})/q \right\rfloor \,\, (j = k+1,\ldots, n)\mbox{,}
\end{equation}
what can be seen as a consequence of proposition \eqref{Prop1}. Hence, in order to decode the received vector
$\bm{r}$, it is not necessary to enumerate all lattice points inside the Lee sphere above-cited, which allows us to
discard many points during the enumeration step by choosing the exact path of the sphere decoding tree \cite{Babak}
that leads to the minimum norm value, given the first $k$ nodes (Figure \eqref{fig:tree}).

The lattice points tested by the algorithm are those whose
coordinate vector $\bm{x}$ satisfies $l_j \leq x_j \leq u_j\mbox{, }
(j = 1, \ldots, k)$, where
\begin{equation}
\begin{split}
l_j &= \left\lceil R + r_j - \sum_{i=1}^{j-1} \left| r_i - x_i
\right| \right\rceil \mbox{ and } \\ u_j &= \left\lfloor -R + r_j +
\sum_{i=1}^{j-1} \left| r_i - x_i \right| \right\rfloor
\label{eq:lims}
\end{split}
\end{equation}
\noindent i.e., the points satisfying $\left\|\bm{x^1} -
(r_1,\ldots,r_k) \right\| \leq R$. In this case, the number of
feasible points corresponds to the number of $\mathbb{Z}^k$ points
inside a Lee sphere of radius $R$ centered at $\bm{r}$, which we
estimate by the volume of the sphere, i.e., $R^k 2^k/k!$. There is a
subtle difference between feasible points and nodes visited by the
Lee sphere decoding algorithm which will become clear later. If we
continue the search until depth $n$ we will get an estimated number
of feasible points as $R^n 2^n/n!$, which, for $n$ much larger than
$k$ represents a drastic reduction. We can now describe the
algorithm of the search done in a node at depth $j \leq k$ as the
enumeration of all $\mathbb{Z}^k$ in the Lee sphere centered at
$(e_1,\ldots,e_k)$ with radius $R$. For $j = k+1$ we choose $x_j$
according to Equation \eqref{minimizadores} and check if
$\left\|M\bm{x} - {\bm r} \right\| \leq R$. In order to speed up the
search some backtracking strategies for updating the decoding radius
are also possible, but we will not consider those in our discussion
on the complexity of the algorithm.

\begin{rmk} We remind that in the classic sphere decoding there are no restrictions on the generator matrix $M$ in order to perform enumerations, since it is possible to triangularize $M$, for example, via QR factorization where $Q$ is an orthogonal matrix. Unfortunately this approach cannot be employed here since rotations are not isometries in the Lee metric. Thus, the ``systematic'' form  \eqref{eq:Geradora} of the generator matrix for $\Lambda_q(C)$ is crucial in the process above-described.
\end{rmk}



\subsection{Choosing the decoding radius}

The radius choice is a critical part of sphere decoding. For the
Euclidean case, Viterbo and Biglieri \cite{Viterbo1} first proposed
the covering radius of a lattice, which can be estimated by Roger's
bound. Hassib and Vikalo \cite{Babak} suggested that the radius
could be chosen accordingly to the signal-to-noise ratio (SNR) of
the channel. Another possible strategy is the so-called Babai's
estimate which can be easily adapted to the Lee norm. The stategy is to take:
\begin{equation}
\hat{R} = \left\| M \left\lfloor \bm{x} \right\rceil - \bm{r}
\right\| \label{eq:BabaiEstimate}
\end{equation}
\noindent where $\bm{x}$ is the (real) solution of $M\bm{x} =
\bm{r}$ and $\left\lfloor \bm{x} \right\rceil$ is the vector whose
coordinates are $\bm{x}$'s entries rounded off to the closest
integer. This estimate guarantees at least one lattice point inside
the Lee sphere of radius $\hat{R}$ and allows us to take advantage
of the interesting structure of $q$-ary lattices. For this estimate and matrix M as in \eqref{eq:Geradora}, we have:

\begin{figure}[!t]
\centering
\includegraphics{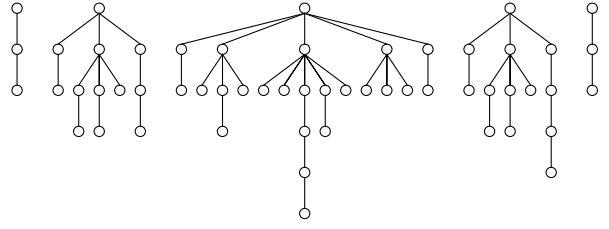}
\caption{Lee sphere decoding tree for $\Lambda_{q}(C)$ where C is the BCH code in Example \eqref{exemplo:BCH}, ${\bm r}=(1,1,1,5,2,3,5)^{t}$, and $R = 2$.}
\label{fig:tree}
\end{figure}

\begin{equation}
\hat{R} \leq \frac{k}{2} + \frac{q(n-k)}{2}.
\end{equation}

Therefore, if $\hat{n}_j$ is the number of visited nodes at
depth-$j$ (corresponding to the number of $\mathbb{Z}^k$ points
inside the ball $\left\|({x_1-r_1,\ldots,x_j-r_j})\right\| \leq
\hat{R}$), we have the following equation as an upper bound for the
expected number of nodes visited by the algorithm until depth $k$:

\begin{equation}
E[\mbox{\# of nodes}] = \sum_{j=0}^k \hat{n}_j \leq\sum_{j=0}^k \frac{(j+q(n-j))^j}{j!}.
\end{equation}

In fact, reasoning in the same way as \cite{Babak}, we can argue that $\hat{R}
\approx k^{1+1/k}/2e$ for large $k$ and hence the expected complexity of the algorithm is exponential,
which is inherent to problem itself.

In the special case that the received vector is in $\mathbb{Z}^n$ (or at least the first $k$ coordinates of
$\bm{r}$ are integers) we have the following:

\begin{prop}
Suppose the received vector $\bm{r}$ is such that $(r_1,\ldots,r_k)
\in \mathbb{Z}^k$. Then the number of nodes of the Lee sphere
decoding tree until depth $k$ is \textit{exactly}

\begin{equation}
\sum_{j = 0}^k \sum_{i=0}^{\min\{j,R\}} 2^i {j \choose i} {R \choose i}.
\end{equation}
\end{prop}
\begin{proof} The proof comes from previous arguments of this section and the fact that the number of points of
$\mathbb{Z}^j$ inside a Lee sphere of radius $R$ is \cite{Golomb}:

\begin{equation}
\sum_{i=0}^{\min\{j,R\}} 2^i {j \choose i} {R \choose i}.
\end{equation}

\end{proof}

\subsection{Comparisons}

There are several efficient algorithms to solve (exactly or
approximately) the Euclidean version of CVP. A first approach to
approximately solve the Lee sphere decoding problem could be through
the so-called Nearest Plane Algorithm \cite{Mic2} which essentially
projects the target vector on a LLL reduced basis for the lattice.
This approach yields a polynomial time algorithm with an exponential
approximate factor. If used to approximate CVP in the Lee metric,
the Nearest Plane Algorithm outputs a vector which satisfies:

\begin{equation} \left\| y - r \right\| \leq \frac{2}{\sqrt{n}} {\left(\frac{2}{\sqrt{3}}\right)}^n \left\|
\bar{y} - r \right\| \end{equation}

\noindent where $\bar{y} \in \Lambda$ is the closest point to $r$ in
the Lee norm and the $1/ \sqrt{n}$ factor is, of course, explained
by the equivalence relation between the $l_1$ and $l_2$ norms.

Concerning the comparison with the classic sphere decoding, we will
not go into detail on the number of arithmetic operations performed
by the algorithms, and let this more careful analysis for a further
work. However, since the performance of the algorithm is closely
related to the volume of the spheres involved in the process, it is
worth to study whether the Lee sphere has a smaller volume than the
Euclidean one, given a received point and its Babai estimate (in
both norms). In what follows we show that when the dimension ($k$)
increases, the Euclidean spheres have greater volume than the Lee
spheres, in average.

Stating the problem more formally, let $\bm{r} = M\bm{x} + \bm{e}$
be a received point and $\hat{R}_1$ and $\hat{R}_2$ the Babai's
estimate to the decoding radius in Lee and Euclidean norm,
respectively. Clearly $\hat{R}_1 \geq \hat{R}_2$. We want to know
whether $\mbox{Vol}(B_{Lee}(\hat{R}_1)) \geq \mbox{Vol}(B_{Euclid}(\hat{R}_2))$ in
average or not, where $\mbox{Vol(S)}$ stands for the Euclidean volume of a set $S$. Without lost of generality we assume that the
transmitted point is the origin. If we fix the value $\hat{R}_2$ and
take the average volume of all Lee spheres centered at the origin
and containing a point of the surface of the Euclidean sphere of
radius $\hat{R}_2$, we have:

\begin{equation}
\begin{split}
&\overline{\mbox{Vol}}(B_{Lee}) = \frac{\int \ldots \int_S V_{Lee}({\phi}_1, \ldots, {\phi}_{n-1}) d{\phi}_1 \ldots {\phi}_{n-1}}{(\pi/2)^{n-1}} \\ &= \frac{\hat{R}_2^n 2^n}{n! (\pi/2)^{n-1}} \int \ldots \int_S (x_1+\ldots+x_n)^n d{\phi}_1 \ldots {\phi}_{n-1} \\
\end{split}
\end{equation}

\noindent where $(x_1,\ldots,x_n)$ is in the surface $S$ of the
Euclidean sphere and the angles $\phi_1, \ldots, \phi_n$ are the
hyperspherical coordinates. If we define

\begin{equation}
I(n,j) := \int \ldots \int_S (x_1+\ldots+x_n)^j d{\phi}_1 \ldots
{\phi}_{n-1},
\end{equation}

\noindent the following expressions can be derived:

\begin{equation}
I(n,n) = \sum_{j=0}^n {n \choose j} \frac{\Gamma(\frac{j+1}{2}) \Gamma(\frac{n-j+1}{2})}{2\Gamma(\frac{n}{2}+1)} I(n-1,j) \mbox{ and}
\end{equation}

\begin{equation}
\label{eq:LeeAverage} \overline{\mbox{Vol}}(B_{Lee}) =
\frac{\hat{R}_2^n 2^n}{n! (\pi/2)^{n-1}}  I(n,n).
\end{equation}

\noindent We can then show that

\begin{equation}
\label{eq:LeeEuclid}
\lim_{n-> \infty} \frac{I(n,n) 2^n/(n! (\pi/2)^{n-1})}{\pi^{n/2}/\Gamma(n/2+1)}
 = 0,
\end{equation}

\noindent what means there is a value $n_o$ such that for all $n
\geq n_o$, we have $\overline{\mbox{Vol}}(B_{Lee}) <
\mbox{Vol}(B_{Euclid}(\hat{R}_2))$. We illustrate this fact in
Figure \ref{fig:teoLee}.

\begin{figure}[!h]

\includegraphics[scale=0.52]{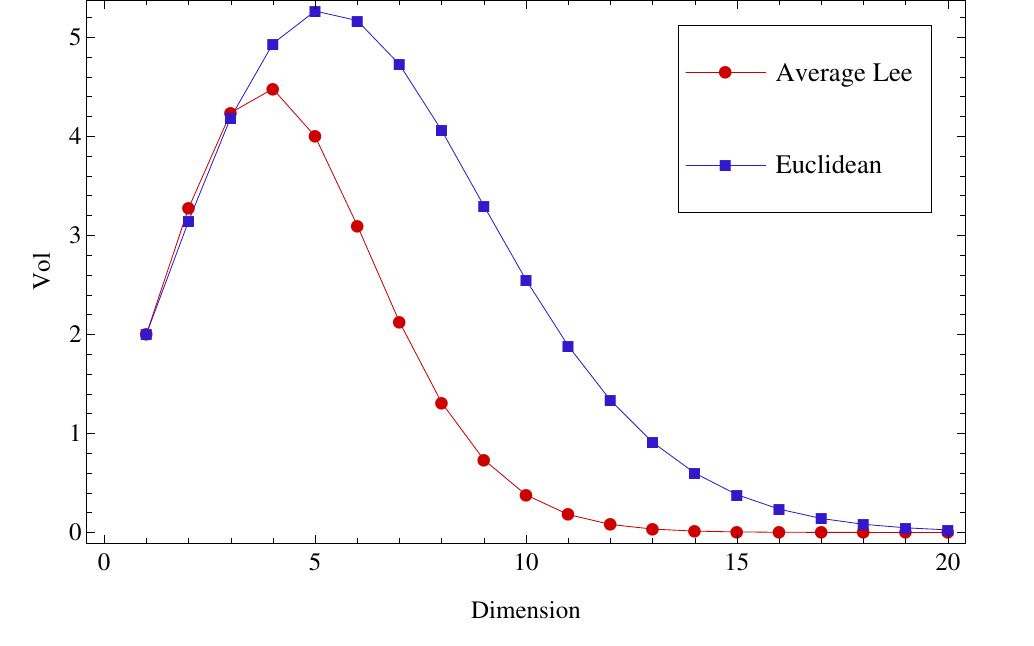}
\caption{Volume of the unitary Euclidean sphere and average volume of Lee sphere
(Equation \eqref{eq:LeeAverage} for $\hat{R}_2 = 1$) versus dimension.}
\label{fig:teoLee}
\end{figure}

\begin{rmk} It is a well-known fact that the ratio between the volume of a  sphere in the $l_1$
norm and a Euclidean sphere of the same radius vanishes while increasing the dimension. This fact, however, does not imply Equation \eqref{eq:LeeEuclid}, since the spheres considered here have different radius.
\end{rmk}

\subsection{Simulations}

To simulate what was proposed in the previous sections we considerer
received vectors of the form:

\begin{equation}
\label{eq:Recebido}
\bm{r} = M\bm{x} + \bm{e}
\end{equation}

\noindent where $M$ is in the form \eqref{eq:Geradora}, the entries
of its submatrix $B$ are uniform on $\mathbb{Z}_q^n$ and the entries
of $e$ are i.i.d. zero mean random variables with Laplace
distribution. Our choice of this noise instead of the usual Gaussian
noise is explained by the relation of Laplace distribution with the
$l_1$ norm. Indeed, while Gaussian noise samples are Euclidean
spherical distributed around the transmitted point, Laplacian noise
samples are Lee spherical distributed. Channels with Laplacian noise
have been investigated in some works (e.g. \cite{Verdu}) as a case
of general channels with additive noise. Figure \ref{fig:sim} shows
simulation in dimensions up to $17$.

\begin{figure}[!h]
\centering
\includegraphics[scale=0.6]{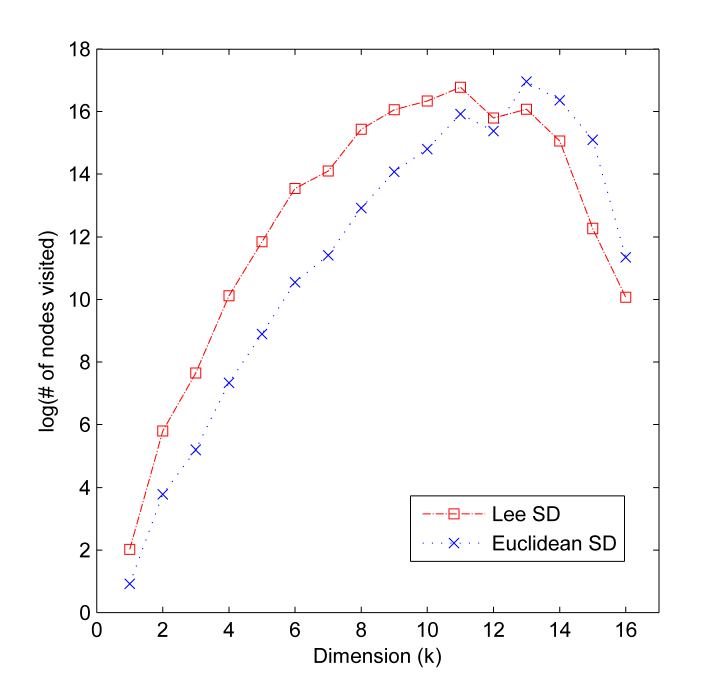}
\caption{Simulation results for fixed $n = 17$, $q = 5$ and $k$ from
$1$ to $16$.} \label{fig:sim}
\end{figure}


\section{Conclusion}

Connections between the decoding process on codes and lattices may
provide tools for error correcting codes and cryptographic schemes.
We discuss here this connection in the case of $q$-ary lattices,
which are obtained from $q$-ary linear block codes through
Construction A, considered with the Lee distance, and present a Lee
sphere decoding algorithm for lattices. Extensions of the presented
approach here to other constructions of lattices will be considered
in a future work.





%

\end{document}